\documentclass[graybox]{svmult}

\usepackage{makeidx}         
\usepackage{graphicx}        
\usepackage{multicol}        
\usepackage[bottom]{footmisc}

\usepackage[english]{babel}
\usepackage{amssymb,amstext,amsmath}

\makeindex             

\begin{document}

\title*{A logistic regression analysis approach for sample survey data based
on phi-divergence measures}
\titlerunning{A logistic regression analysis approach for sample survey data}
\author{Elena Castilla, Nirian Mart\'{i}n \and
Leandro Pardo\inst{*}}
\institute{ \inst{*}Complutense University of Madrid, 28040 Madrid, Spain \texttt{lpardo@mat.ucm.es} 
\begin{quotation}
I had the great honor to have Pedro as a friend; a deeply beloved friend. We have shared not only many scientific interests in connection with Statistical Information Theory but also, in accordance with our frequent conversations, many similarities in our youthful life experiences. The research in this paper, written with two coauthors, is dedicated in memoriam to Professor Pedro Gil.
\end{quotation}
}
%
%
\maketitle

\abstract{A new family of minimum distance estimators for binary logistic regression models based on $\phi$-divergence measures is introduced. The so called ``pseudo minimum phi-divergence estimator''(PM$\phi$E) family is presented as an extension of `` minimum phi-divergence estimator'' (M$\phi$E) for general sample survey designs and contains, as a particular case, the pseudo maximum likelihood estimator (PMLE) considered in Roberts et al. \cite{r}. Through a simulation study it is shown that some PM$\phi$Es have a better behaviour, in terms of efficiency, than the PMLE.\\ }


\abstract*{A new family of minimum distance estimators for binary logistic regression models based on $\phi$-divergence measures is introduced. The so called ``pseudo minimum phi-divergence estimator''(PM$\phi$E) family is presented as an extension of M$\phi$E for general sample survey designs and contains, as a particular case, the pseudo maximum likelihood estimator (PMLE) considered in Roberts et al.(1987). Through a simulation study it is shown that some PM$\phi$Es have a better behaviour, in terms of efficiency, than the PMLE.}

\section{Introduction}
\label{EleNirLea_sec:1}

Suppose that the population of interest is partitioned into $I$ cells or
domains according to the levels of one or more factors. Let $N_{i}$
($i=1,...,I$) denote the $i$-th domain size, $N=%
{\textstyle\sum_{i=1}^{I}}
N_{i}$ the population domain total and $N_{i1}$, the population counts, out of
$N_{i}$, where the binary response ($0$ for failure and $1$ for success)
variable is equal to $1$. Since $N_{i1}$ and $N_{i}$\ are fixed but unknown
values ($i=1,...,I$), $\widehat{N}_{i}$ denotes the survey estimator of the
$i$-th domain size $N_{i}$ and $\widehat{N}_{i1}$ the corresponding estimate
of the successful events $N_{i1}$. The ratio estimator $\widehat{p}%
_{i}=\widehat{N}_{i1}/\widehat{N}_{i},i=1,...,I$, is often used to estimate
the population proportion of successful events, $\pi_{i}=\frac{N_{i1}}{N_{i}%
},$ $i=1,...,I.$ Standard sampling theory provides an estimator of the
covariance matrix of the $\widehat{\boldsymbol{p}}=\left(  \widehat{p}%
_{1},...,\widehat{p}_{I}\right)  ^{T}$. Another choice is using the logistic
regression
\begin{equation}
\pi\left(  \boldsymbol{x}_{i}^{T}\boldsymbol{\beta}\right)  =\frac
{\exp\{\boldsymbol{x}_{i}^{T}\boldsymbol{\beta}\}}{1+\exp\{\boldsymbol{x}%
_{i}^{T}\boldsymbol{\beta}\}}=\frac{\exp\{\beta_{0}+%
{\textstyle\sum\limits_{s=1}^{k}}
\beta_{j}x_{ij}\}}{1+\exp\{\beta_{0}+%
{\textstyle\sum\limits_{s=1}^{k}}
\beta_{j}x_{ij}\}},\text{ }i=1,...,I,  \label{EleNirLea_eq:0} 
\end{equation}
to modelize the population proportion of successful events, $$\pi_{i}%
=\pi\left(  \boldsymbol{x}_{i}^{T}\boldsymbol{\beta}\right)  =\frac{N_{i1}%
}{N_{i}},$$ which is assumed to depend on constants $x_{ij}$, $j=1,...,k$
($k<I$) derived from the factor levels, summarized in a ($k+1$)-vector of
known constants $\boldsymbol{x}_{i}=\left(  1,x_{i1},...,x_{ik}\right)  ^{T}$,
and also on a ($k+1$)-vector of parameters $\boldsymbol{\beta}=(\beta
_{0},\beta_{1},...,\beta_{k})^{T}$.

Under independent binomial sampling in each domain, it is well-known that the
maximum likelihood estimator (MLE) of $\boldsymbol{\beta}$,
$\widehat{\boldsymbol{\beta}}$, is obtained through iterative calculations
from the following likelihood equations
\begin{equation}
\boldsymbol{X}^{T}\mathrm{diag}(\boldsymbol{n})\boldsymbol{\pi}\left(
\boldsymbol{\beta}\right)  =\boldsymbol{X}^{T}\mathrm{diag}(\boldsymbol{n}%
)\widehat{\boldsymbol{q}}, \label{EleNirLea_eq:0.1}%
\end{equation}
where $\boldsymbol{X}\negthinspace=(\boldsymbol{x}_{1},...,\boldsymbol{x}_{I})^{T}$ is a
full rank matrix, $\boldsymbol{\pi}\left(  \boldsymbol{\beta}\right)\negthinspace  =\left(
\pi\left(  \boldsymbol{x}_{1}^{T}\boldsymbol{\beta}\right)  ,...,\pi\left(
\boldsymbol{x}_{I}^{T}\boldsymbol{\beta}\right)  \right)^{T}\negthinspace$ ,
$\widehat{\boldsymbol{q}}=\left(  \widehat{q}_{1},...,\widehat{q}_{I}\right)
^{T}$ with $\widehat{q}_{i}=n_{i1}/n_{i}$, $n_{i}$ being the sample size from
the $i$-th domain, $n=%
{\textstyle\sum\nolimits_{i=1}^{I}}
n_{i}$ the $i$-th sample domain total and $n_{i1}$ the sample total of
successful events the $i$-th domain. If we consider the probability vectors%
\begin{align*}
\widehat{\boldsymbol{p}}^{\ast}  &  =\left(  \frac{n_{1}}{n}\widehat{q}%
_{1},\frac{n_{1}}{n}\left(  1-\widehat{q}_{1}\right)  ,...,\frac{n_{I}}%
{n}\widehat{q}_{I},\frac{n_{I}}{n}\left(  1-\widehat{q}_{I}\right)  \right)
^{T}\\
&  =\left(  \frac{n_{11}}{n},\frac{n_{1}-n_{11}}{n},...,\frac{n_{I1}}{n}%
,\frac{n_{I}-n_{I1}}{n}\right)  ^{T}\\
&  =\left(\frac{n_{11}}{n},\frac{n_{12}}{n},...,\frac{n_{I1}}{n},\frac{n_{I2}}%
{n}\right)^{T},\text{ \qquad}(n_{i2}=n_{i}-n_{i1}),
\end{align*}
and
 \begin{eqnarray*}
\boldsymbol{p}^{\ast}\negthinspace\left(  \boldsymbol{\beta}\right)\negthinspace=\left(  \textstyle\frac{n_{1}%
}{n}\pi\left(  \boldsymbol{x}_{1}^{T}\boldsymbol{\beta}\right) \negthinspace ,\frac{n_{1}%
}{n}\left(  1-\pi\left(  \boldsymbol{x}_{1}^{T}\boldsymbol{\beta}\right)
\right)  ,\dots,\frac{n_{I}}{n}\pi\left(  \boldsymbol{x}_{I}^{T}%
\boldsymbol{\beta}\right)\negthinspace  ,\frac{n_{I}}{n}\left(  1-\pi\left(  \boldsymbol{x}%
_{I}^{T}\boldsymbol{\beta}\right)  \right)  \right)  ^{T}%
 \end{eqnarray*}
the MLE of\ $\boldsymbol{\beta}$, $\widehat{\boldsymbol{\beta}}$, can be
equivalently defined by
\[
\widehat{\boldsymbol{\beta}}=\arg\min_{\boldsymbol{\beta\in%
\mathbb{R}
}^{k+1}}d_{Kullback}\left(  \widehat{\boldsymbol{p}}^{\ast},\boldsymbol{p}%
^{\ast}\left(  \boldsymbol{\beta}\right)  \right)  ,
\]
where \ $d_{Kullback}\left(  \widehat{\boldsymbol{p}}^{\ast},\boldsymbol{p}%
^{\ast}\left(  \boldsymbol{\beta}\right)  \right)  $ is the Kullback-Leibler
divergence between the probability vectors $\widehat{\boldsymbol{p}}^{\ast}$
and $\boldsymbol{p}^{\ast}\left(  \boldsymbol{\beta}\right)  $ defined by
\[
d_{Kullback}\left(  \widehat{\boldsymbol{p}}^{\ast},\boldsymbol{p}^{\ast
}\left(  \boldsymbol{\beta}\right)  \right)\negthinspace  =\sum_{i=1}^{I}\left[
\frac{n_{i1}}{n}\log\frac{n_{i1}}{n_{i}\pi\left(  \boldsymbol{x}_{i}%
^{T}\boldsymbol{\beta}\right)  }+\frac{n_{i2}}{n}\log\frac{n_{i2}}%
{n_{i}\left(  1-\pi\left(  \boldsymbol{x}_{i}^{T}\boldsymbol{\beta}\right)
\right)  }\right] \negthinspace .
\]
In Pardo et al. \cite{p3} the
minimum phi-divergence estimator (M$\phi$E) was introduced, as a natural extension of the MLE, as%
\begin{equation}
\widehat{\boldsymbol{\beta}}_{\phi}=\arg\min_{\boldsymbol{\beta\in%
\mathbb{R}
}^{k+1}}d_{\phi}\left(  \widehat{\boldsymbol{p}}^{\ast},\boldsymbol{p}^{\ast
}\left(  \boldsymbol{\beta}\right)  \right),  \label{EleNirLea_eq:0.21}%
\end{equation}
where $d_{\phi}\left(  \widehat{\boldsymbol{p}}^{\ast},\boldsymbol{p}^{\ast
}\left(  \boldsymbol{\beta}\right)  \right)  $ is the phi-divergence measure
between the probability vectors $\widehat{\boldsymbol{p}}^{\ast}$ and
$\boldsymbol{p}^{\ast}\left(  \boldsymbol{\beta}\right)  $ given by
 \begin{eqnarray*}
d_{\phi}\left(  \widehat{\boldsymbol{p}}^{\ast},\boldsymbol{p}^{\ast}\left(
\boldsymbol{\beta}\right)  \right)  =&\displaystyle \sum_{i=1}^{I}\frac{n_{i}}{n}\left[
\pi\left(  \boldsymbol{x}_{i},\boldsymbol{\beta}\right)  \phi\left(
\tfrac{n_{i1}}{n_{i}\pi\left(  \boldsymbol{x}_{i}^{T}\boldsymbol{\beta
}\right)  }\right)\right.\\  + &  \left. \left(  1-\pi\left(  \boldsymbol{x}_{i},\boldsymbol{\beta
}\right)  \right)  \phi\left(  \tfrac{n_{i2}}{n_{i}\left(  1-\pi\left(
\boldsymbol{x}_{i}^{T}\boldsymbol{\beta}\right)  \right)  }\right)  \right]
,
\end{eqnarray*}
with $\phi\in\Phi^{\ast}$. By $\Phi^{\ast}$ we are denoting the class of all
convex functions, $\phi\left(  x\right)  $, $x>0$, such that at $x=1$,
$\phi\left(  1\right)  =\phi^{\prime}\left(  1\right)  =0$, and at $x=0$,
$0\phi\left(  0/0\right)  =0$ and $0\phi\left(  p/0\right)  =p\lim
_{u\rightarrow\infty}\frac{\phi\left(  u\right)  }{u}$. For every $\phi\in
\Phi^{\ast}$, differentiable at $x=1$, the function
\[
\Psi\left(  x\right)  =\phi\left(  x\right)  -\phi^{\prime}\left(  1\right)
\left(  x-1\right)
\]
also belongs to $\Phi^{\ast}$. Therefore, we have $d_{\psi}\left(
\widehat{\boldsymbol{p}}^{\ast},\boldsymbol{p}^{\ast}\left(  \boldsymbol{\beta
}\right)  \right)  =d_{\phi}\left(  \widehat{\boldsymbol{p}}^{\ast
},\boldsymbol{p}^{\ast}\left(  \boldsymbol{\beta}\right)  \right)  $ and
$\psi$ has the additional property that $\psi^{\prime}(1)=0$. Since the two
divergence measures are equivalent, we can consider the set $\Phi^{\ast}$ to
be equivalent to the set
\[
\Phi=\Phi^{\ast}\cap\left\{  \phi:\phi^{\prime}(1)=0\right\}  .
\]
For more details  see Cressie and Pardo \cite{p1} and Pardo \cite{p2}. In what
follows, we give our theoretical results for $\phi\in\Phi$, but often apply
them to choices of functions in $\Phi^{\ast}$.

For general sample survey designs we do not have maximum likelihood estimators
due to difficulties in obtaining appropriate likelihood functions. Hence, it
is a common practice to use a pseudo maximum likelihood estimator (PMLE) of
$\boldsymbol{\beta}$, $\widehat{\boldsymbol{\beta}}_{P}$, obtained from
(\ref{EleNirLea_eq:0.1}) by replacing $n_{i}/n$, $i=1,...,I$, by the estimated domain
relative size $w_{i}=\widehat{N}_{i}/\widehat{N}$, $i=1,...,I,$ and the sample
proportions $\widehat{q}_{i}=n_{i1}/n_{i},$ $i=1,...,I,$ by the ratio estimate
$\widehat{p}_{i}=\widehat{N}_{i1}/\widehat{N}_{i},i=1,...,I,$%
\begin{equation}
\boldsymbol{X}^{T}\mathrm{diag}(\boldsymbol{w})\boldsymbol{\pi}\left(
\boldsymbol{\beta}\right)  =\boldsymbol{X}^{T}\mathrm{diag}(\boldsymbol{w}%
)\widehat{\boldsymbol{p}}, \label{EleNirLea_eq:0.3}%
\end{equation}
where $\boldsymbol{w}=(w_{1},...,w_{I})^{T}$.

In this paper we extend the concept of M$\phi$E by considering the
\textquotedblleft pseudo minimum phi-divergence estimator\textquotedblright%
\ (PM$\phi$E) as a natural extension of the PMLE\ and we solve some
statistical problem for the model considered in (\ref{EleNirLea_eq:0}). In Section
\ref{EleNirLea_sec:2} we shall introduce the PM$\phi$E\ for general sample designs and we
study its asymptotic behavior. A numerical example is presented in Section \ref{EleNirLea_sec:3} and, finally, in Section \ref{EleNirLea_sec:4} a simulation study is carried out.

\section{Pseudo minimum phi-divergence estimator for general sample designs}
\label{EleNirLea_sec:2}

For general sample designs, we should consider the kernel of the weighted
loglikelihood%
\[
\ell_{\boldsymbol{w}}\left(  \boldsymbol{\beta}\right)  =n\sum_{i=1}^{I}%
w_{i}\left[  \widehat{p}_{i}\log\pi\left(  \boldsymbol{x}_{i}^{T}%
\boldsymbol{\beta}\right)  +(1-\widehat{p}_{i})\log(1-\pi\left(
\boldsymbol{x}_{i}^{T}\boldsymbol{\beta}\right)  )\right]  ,
\]
which is derived from the kernel of the likelihood for $I$ independent binomial
random variables
\begin{align*}
\ell\left(  \boldsymbol{\beta}\right)   &  =\sum_{i=1}^{I}n_{i}\left[
\widehat{q}_{i}\log\pi\left(  \boldsymbol{x}_{i}^{T}\boldsymbol{\beta}\right)
+(1-\widehat{q}_{i})\log(1-\pi\left(  \boldsymbol{x}_{i}^{T}\boldsymbol{\beta
}\right)  )\right] \\
&  =n\sum_{i=1}^{I}\frac{n_{i}}{n}\left[  \widehat{q}_{i}\log\pi\left(
\boldsymbol{x}_{i}^{T}\boldsymbol{\beta}\right)  +(1-\widehat{q}_{i}%
)\log(1-\pi\left(  \boldsymbol{x}_{i}^{T}\boldsymbol{\beta}\right)  )\right]
,
\end{align*}
replacing $\frac{n_{i}}{n}$ by $w_{i}=\widehat{N}_{i}/\widehat{N}$, and
$\widehat{q}_{i}=n_{i1}/n_{i}$ by $\widehat{p}_{i}=\widehat{N}_{i1}%
/\widehat{N}_{i}$, $i=1,...,I$. If we consider the two probability vectors%
\[
\widehat{\boldsymbol{p}}_{\boldsymbol{w}}=\left(  w_{1}\widehat{p}_{1}%
,w_{1}\left(  1-\widehat{p}_{1}\right)  ,...,w_{I}\widehat{p}_{I},w_{I}\left(
1-\widehat{p}_{I}\right)  \right)  ^{T}%
\]
and
\[
\boldsymbol{p}_{\boldsymbol{w}}\negthinspace\left(  \boldsymbol{\beta}\right)\negthinspace  =\left(
w_{1}\pi\left(  \boldsymbol{x}_{1}^{T}\boldsymbol{\beta}\right)\negthinspace  ,w_{1}\left(
1-\pi\left(  \boldsymbol{x}_{1}^{T}\boldsymbol{\beta}\right)  \right)
,...,w_{I}\pi\left(  \boldsymbol{x}_{I}^{T}\boldsymbol{\beta}\right)
\negthinspace,w_{I}\left(  1-\pi\left(  \boldsymbol{x}_{I}^{T}\boldsymbol{\beta}\right)
\right)  \right)\negthinspace  ^{T}\negthinspace,
\]
we get%
\[
\ell_{\boldsymbol{w}}\left(  \boldsymbol{\beta}\right)  =-nd_{Kullback}\left(
\widehat{\boldsymbol{p}}_{\boldsymbol{w}},\boldsymbol{p}_{\boldsymbol{w}%
}\left(  \boldsymbol{\beta}\right)  \right)  +k,
\]
where $k$ is a constant not depending on $\boldsymbol{\beta}$. Therefore the
PMLE of $\boldsymbol{\beta}$, $\widehat{\boldsymbol{\beta}}_{P}$, presented in
(\ref{EleNirLea_eq:0.3}) can be defined as
\[
\widehat{\boldsymbol{\beta}}_{P}=\arg\max_{\boldsymbol{\beta\in%
\mathbb{R}
}^{k+1}}\ell_{\boldsymbol{w}}\left(  \boldsymbol{\beta}\right)  =\arg
\min_{\boldsymbol{\beta\in%
\mathbb{R}
}^{k+1}}d_{Kullback}\left(  \widehat{\boldsymbol{p}}_{\boldsymbol{w}%
},\boldsymbol{p}_{\boldsymbol{w}}\left(  \boldsymbol{\beta}\right)  \right)
.
\]
Based on the previous interpretation of the PMLE, in the following definition
we shall present the PM$\phi$E.

\begin{definition}
The PM$\phi$E in a general sample design for the parameter $\boldsymbol{\beta
}$ in the model considered in (\ref{EleNirLea_eq:0}) is defined as
\[
\widehat{\boldsymbol{\beta}}_{\phi,P}=\arg\min_{\boldsymbol{\beta\in%
\mathbb{R}
}^{k+1}}d_{\phi}\left(  \widehat{\boldsymbol{p}}_{\boldsymbol{w}%
},\boldsymbol{p}_{\boldsymbol{w}}\left(  \boldsymbol{\beta}\right)  \right)
,
\]
where%
 \begin{eqnarray*}
d_{\phi}\left(  \widehat{\boldsymbol{p}}_{\boldsymbol{w}},\boldsymbol{p}%
_{\boldsymbol{w}}\left(  \boldsymbol{\beta}\right)  \right)  =&\displaystyle \sum_{i=1}%
^{I}w_{i}\left[  \pi\left(  \boldsymbol{x}_{i},\boldsymbol{\beta}\right)
\phi\left(  \tfrac{\widehat{p}_{i}}{\pi\left(  \boldsymbol{x}_{i}%
^{T}\boldsymbol{\beta}\right)  }\right)\right.\\ + & \left. \left(  1-\pi\left(  \boldsymbol{x}%
_{i},\boldsymbol{\beta}\right)  \right)  \phi\left(  \tfrac{1-\widehat{p}_{i}%
}{1-\pi\left(  \boldsymbol{x}_{i}^{T}\boldsymbol{\beta}\right)  }\right)
\right]
 \end{eqnarray*}
is the phi-divergence measure between the probability vectors
$\widehat{\boldsymbol{p}}_{\boldsymbol{w}}$ and $\boldsymbol{p}%
_{\boldsymbol{w}}\left(  \boldsymbol{\beta}\right)  $.
\end{definition}

The following result establishes the asymptotic distribution of the PM$\phi$E
of $\boldsymbol{\beta}$, $\widehat{\boldsymbol{\beta}}_{\phi,P}$.

\begin{theorem}
Let us assume that $\boldsymbol{\beta}_{0}$ is the true value of
$\boldsymbol{\beta}$ and%
\begin{align*}
&  \boldsymbol{w}\underset{n\rightarrow\infty}{\overset{p}{\longrightarrow}%
}\boldsymbol{W},\qquad\boldsymbol{W}=(W_{1},...,W_{I})^{T},\qquad W_{i}%
=\frac{N_{i}}{N},\\
&  \widehat{\boldsymbol{p}}\underset{n\rightarrow\infty
}{\overset{p}{\longrightarrow}}\boldsymbol{\pi}\left(  \boldsymbol{\beta}%
_{0}\right)  ,\qquad\sqrt{n}(\widehat{\boldsymbol{p}}-\boldsymbol{\pi}\left(
\boldsymbol{\beta}_{0}\right)  )\underset{n\rightarrow\infty
}{\overset{\mathcal{L}}{\longrightarrow}}\mathcal{N}\left(  \boldsymbol{0}%
,\boldsymbol{V}\right)  .
\end{align*}
Then, we have
\[
\sqrt{n}(\widehat{\boldsymbol{\beta}}_{\phi,P}-\boldsymbol{\beta}%
_{0})\underset{n\rightarrow\infty}{\overset{\mathcal{L}}{\longrightarrow}%
}\mathcal{N}\left(  \boldsymbol{0}_{k+1},\boldsymbol{V}\left(
\boldsymbol{\beta}_{0}\right)  \right)  ,
\]
where%
\begin{align}
\boldsymbol{V}\left(  \boldsymbol{\beta}_{0}\right)   &  =(\boldsymbol{X}%
^{T}\boldsymbol{\Delta X})^{-1}\boldsymbol{X}^{T}\mathrm{diag}\left(
\boldsymbol{W}\right)  \boldsymbol{V}\mathrm{diag}\left(  \boldsymbol{W}%
\right)  \boldsymbol{X}(\boldsymbol{X}^{T}\boldsymbol{\Delta X})^{-1}%
,\label{EleNirLea_eq:var}\\
\boldsymbol{\Delta} &  =\mathrm{diag}\{W_{i}\pi\left(  \boldsymbol{x}_{i}%
^{T}\boldsymbol{\beta}_{0}\right)  \left(  1-\pi\left(  \boldsymbol{x}_{i}%
^{T}\boldsymbol{\beta}_{0}\right)  \right)  \}_{i=1,...,I}.\nonumber
\end{align}

\end{theorem}

\begin{proof}
Based on Theorem 1 in Pardo et al. \cite{p3}, we have%
\begin{align*}
\widehat{\boldsymbol{\beta}}_{\phi,P} &\negthinspace  =\boldsymbol{\beta}_{0}%
\negthinspace+ \negthinspace(\boldsymbol{X}^{T}\boldsymbol{\Delta X})^{-1}\boldsymbol{X}^{T}%
\mathrm{diag}\left\{  \boldsymbol{c}_{i}^{T}\right\}  _{i=1}^{I}%
\mathrm{diag}^{-1/2}\left(  \boldsymbol{p}_{\boldsymbol{w}}\negthinspace\left(
\boldsymbol{\beta}_{0}\right)  \right)  \left(  \widehat{\boldsymbol{p}%
}_{\boldsymbol{w}}-\boldsymbol{p}_{\boldsymbol{w}}\negthinspace\left(  \boldsymbol{\beta
}_{0}\right)  \right)  \\
&  +o\left(  \left\Vert \mathrm{diag}\left\{  \boldsymbol{c}_{i}^{T}\right\}
_{i=1}^{I}\mathrm{diag}^{-1/2}\left(  \boldsymbol{p}_{\boldsymbol{w}}\negthinspace\left(
\boldsymbol{\beta}_{0}\right)  \right)  \left(  \widehat{\boldsymbol{p}%
}_{\boldsymbol{w}}-\boldsymbol{p}_{\boldsymbol{w}}\negthinspace\left(  \boldsymbol{\beta
}_{0}\right)  \right)  \right\Vert \boldsymbol{1}_{k+1}\right)  ,
\end{align*}
with%
\[
\boldsymbol{c}_{i}=\left(  w_{i}\pi\left(  \boldsymbol{x}_{i}^{T}%
\boldsymbol{\beta}_{0}\right)  \left(  1-\pi\left(  \boldsymbol{x}_{i}%
^{T}\boldsymbol{\beta}_{0}\right)  \right)  \right)  ^{1/2}\left(
\begin{array}
[c]{c}%
\left(  1-\pi\left(  \boldsymbol{x}_{i}^{T}\boldsymbol{\beta}_{0}\right)
\right)  ^{1/2}\\
-\pi\left(  \boldsymbol{x}_{i}^{T}\boldsymbol{\beta}_{0}\right)^{1/2}
\end{array}
\right)  ,\text{ }i=1,..,I.
\]
Since%
\[
\mathrm{diag}\left\{  \boldsymbol{c}_{i}^{T}\right\}  _{i=1}^{I}%
\mathrm{diag}^{-1/2}\left(  \boldsymbol{p}_{\boldsymbol{w}}\left(
\boldsymbol{\beta}_{0}\right)  \right)  \left(  \widehat{\boldsymbol{p}%
}_{\boldsymbol{w}}-\boldsymbol{p}_{\boldsymbol{w}}\left(  \boldsymbol{\beta
}_{0}\right)  \right)  =\mathrm{diag}\left(  \boldsymbol{w}\right)  \left(
\widehat{\boldsymbol{p}}-\boldsymbol{\pi}\left(  \boldsymbol{\beta}%
_{0}\right)  \right)  ,
\]
$\boldsymbol{w}\underset{n\rightarrow\infty}{\overset{p}{\longrightarrow}%
}\boldsymbol{W}$ and $\widehat{\boldsymbol{p}}\underset{n\rightarrow
\infty}{\overset{p}{\longrightarrow}}\boldsymbol{\pi}\left(  \boldsymbol{\beta
}_{0}\right)  $, it holds%
\[
\sqrt{n}(\widehat{\boldsymbol{\beta}}_{\phi,P}-\boldsymbol{\beta}%
_{0})=(\boldsymbol{X}^{T}\boldsymbol{\Delta X})^{-1}\boldsymbol{X}%
^{T}diag\left(  \boldsymbol{W}\right)  \sqrt{n}\left(  \widehat{\boldsymbol{p}%
}-\boldsymbol{\pi}\left(  \boldsymbol{\beta}_{0}\right)  \right)
+o_{p}(\boldsymbol{1}_{k+1}).
\]
From the Sluysky's theorem and taking into account $\sqrt{n}%
(\widehat{\boldsymbol{p}}-\boldsymbol{\pi}\left(  \boldsymbol{\beta}%
_{0}\right)  )\underset{n\rightarrow\infty}{\overset{\mathcal{L}%
}{\longrightarrow}}\mathcal{N}\left(  \boldsymbol{0}_{k+1},\boldsymbol{V}%
\right)  $, it follows the desired result.
\qed
\end{proof}

\begin{remark}
Under independent binomial sampling in each domain, it is well-known that
$\boldsymbol{V}=\mathrm{diag}\{\pi\left(  \boldsymbol{x}_{i}^{T}%
\boldsymbol{\beta}_{0}\right)  \left(  1-\pi\left(  \boldsymbol{x}_{i}%
^{T}\boldsymbol{\beta}_{0}\right)  \right)  \}_{i=1,...,I}\mathrm{diag}%
^{-1}\left(  \boldsymbol{W}\right)  $ and hence $\boldsymbol{V}\left(
\boldsymbol{\beta}_{0}\right)  =(\boldsymbol{X}^{T}\boldsymbol{\Delta X}%
)^{-1}$, which matches Theorem 2 in Pardo et al. \cite{p3}.
\end{remark}

\begin{remark}
The asymptotic results obtained in the current paper differ from Castilla et al. \cite{castilla} in the elements tending to infinite, here the total individuals in
the whole sample, $n$, while in the cited paper is the total number of clusters
what tends to infinite.
\end{remark}

\section{A numerical example}
\label{EleNirLea_sec:3}
In order to obtain the PM$\phi$Es, from a
practical point of view, we can give an explicit expression for $\phi $. In
this paper we shall focus on the Cressie-Read subfamily%
\[
\phi _{\lambda }(x)=\left\{ 
\begin{array}{ll}
\frac{1}{\lambda (1+\lambda )}\left[ x^{\lambda +1}-x-\lambda (x-1)\right] ,
& \lambda \in 
\mathbb{R}
-\{-1,0\} \\ 
\lim_{\upsilon \rightarrow \lambda }\frac{1}{\upsilon (1+\upsilon )}\left[
x^{\upsilon +1}-x-\upsilon (x-1)\right] , & \lambda \in \{-1,0\}%
\end{array}%
\right. . 
\]%
We can observe that for $\lambda =0$, we have 
\[
\phi _{\lambda =0}(x)=\lim_{\upsilon \rightarrow 0}\frac{1}{\upsilon
(1+\upsilon )}\left[ x^{\upsilon +1}-x-\upsilon (x-1)\right] =x\log x-x+1, 
\]%
and the associated phi-divergence, coincides with the Kullback divergence,
therefore the PM$\phi$Es based on $\phi_{\lambda }(x)$ contains as special case the PMLE.

We shall consider the example presented in Molina et al. \cite{molina}. A random
subsample of $50$ clusters (primary sampling units) containing $1299$
households was selected from the 1975 U.K. Family Expenditure Survey. These
households are divided into $12$ groups of sizes $n_{1},...,n_{12}$ by age
of head of household ($4$ levels) and number of persons in the household ($3$
levels). The binary response is $1$ if the household owns the dwelling it
occupies and $0$ otherwise. The number of households $r_{i}$ for which the
binary response is $1$, together with $n_{i}$ are shown in Table 1 of the
cited paper.

We denote by $\beta _{1(r)}$ the parameter associated to the level $r$ of
the factor \textquotedblleft age of head of housholds\textquotedblright , $%
r=2,3$ and $4$ since $\beta _{1(1)}=0$ and by $\beta _{2(s)}$ the parameter
associated to the level $s$ of the factor \textquotedblleft number of
persons in the housholds\textquotedblright , $\ s=2,3,$ since we assume $%
\beta _{2(1)}=0$. The parameter vector with unknown values will be denote by 
\[
\boldsymbol{\beta }=(\beta _{0},\beta _{1(2)},\beta _{1(3)},\beta
_{1(4)},\beta _{2(2)},\beta _{2(3)})^{T}.
\]%
The design matrix that we are going to consider for the example under
consideration is given by 
\[
\boldsymbol{X=}\left( 
\begin{array}{cccccccccccc}
1 & 1 & 1 & 1 & 1 & 1 & 1 & 1 & 1 & 1 & 1 & 1 \\ 
0 & 0 & 0 & 1 & 1 & 1 & 0 & 0 & 0 & 0 & 0 & 0 \\ 
0 & 0 & 0 & 0 & 0 & 0 & 1 & 1 & 1 & 0 & 0 & 0 \\ 
0 & 0 & 0 & 0 & 0 & 0 & 0 & 0 & 0 & 1 & 1 & 1 \\ 
0 & 1 & 0 & 0 & 1 & 0 & 0 & 1 & 0 & 0 & 1 & 0 \\ 
0 & 0 & 1 & 0 & 0 & 1 & 0 & 0 & 1 & 0 & 0 & 1%
\end{array}%
\right) ^{T}=(\boldsymbol{x}_{1},...,\boldsymbol{x}_{12})^{T}
\]%
and the logistic regression model under consideration is given by 
\[
\pi \left( \boldsymbol{x}_{i}^{T}\boldsymbol{\beta }\right) =\frac{\exp
\left\{ \boldsymbol{x}_{i}^{T}\boldsymbol{\beta }\right\} }{1+\exp \left\{ 
\boldsymbol{x}_{i}^{T}\boldsymbol{\beta }\right\} },\text{ }i=1,...,12,
\]%
equivalent to 
\[
\pi \left( \boldsymbol{x}_{i}^{T}\boldsymbol{\beta }\right) =\frac{\exp
\{\beta _{0}+\beta _{1(r)}+\beta _{2(s)}\}}{1+\exp \{\beta _{0}+\beta
_{1(r)}+\beta _{2(s)}\}},
\]%
if the $i$-th probability is associated with the $r$-th level of the first
variable ($r=1,...,4$) and the $s$-th level of the second variable ($s=1,...,3$).
In the following table we present the pseudo minimum Cressie-Read divergence estimators (PMCREs) of $\boldsymbol{\beta}$, $\boldsymbol{\beta}_{\lambda,P}$,
 for $\lambda \in \left\{ 0,2/3,1,2\right\} $.

\begin{table}[h]
\setlength{\tabcolsep}{6pt}
\centering
\caption{PMCREs for the clustered Family Expenditure Survey data model.}
\label{EleNirLea_tab:0}
\begin{tabular}{r cccccc}
\hline
$\lambda$ & $\widehat{\beta}_{0,\lambda,P}$ &$\widehat{\beta}_{1(2),\lambda,P}$ & $\widehat{\beta}_{1(3),\lambda,P}$ & $\widehat{\beta}_{1(4),\lambda,P}$ & $\widehat{\beta}_{2(2),\lambda,P}$       & $\widehat{\beta}_{2(3),\lambda,P}$ \\
\hline
0  & $-$0.1585 & 0.4403 & $-$0.1412 & $-$0.4179 & 0.5042 & 0.4703 \\
2/3  & $-$0.1564 & 0.4291 & $-$0.1436 & $-$0.4174 & 0.4985 & 0.4735 \\
1  & $-$0.1574 & 0.4251 & $-$0.1438 & $-$0.4158 & 0.4971 & 0.476  \\
2  & $-$0.1663 & 0.4192 & $-$0.1408 & $-$0.4075 & 0.4974 & 0.4856 \\
\hline
\end{tabular}
\end{table}

\section{Simulation Study}
\label{EleNirLea_sec:4}

The following simulation study has been designed by following the previous example. Since in the logistic regression model there are
two factors, the first one with $4$ categories and the second one with $3$
categories, in total $I=12$ domains are taken into account. Let%
\[
\boldsymbol{p}\left(  \boldsymbol{\beta}\right)  =\left(  \tfrac{N_{1}}{N}%
\pi\left(  \boldsymbol{x}_{1}^{T}\boldsymbol{\beta}\right)  ,\tfrac{N_{1}}%
{N}\left(  1-\pi\left(  \boldsymbol{x}_{1}^{T}\boldsymbol{\beta}\right)
\right)  ,...,\tfrac{N_{I}}{N}\pi\left(  \boldsymbol{x}_{I}^{T}%
\boldsymbol{\beta}\right)  ,\tfrac{N_{I}}{N}\left(  1-\pi\left(
\boldsymbol{x}_{I}^{T}\boldsymbol{\beta}\right)  \right)  \right)  ^{T}
\]
be the theoretical probability vector in the logistic regression with complex
sampling. The values of the components of $\boldsymbol{p}\left(
\boldsymbol{\beta}\right)  $ are given in Table \ref{EleNirLea_tab:1}. In total
$n=1299$ individuals are taken from the primary units of the sample, $J=50$
clusters, of size $m_{(j)}=26$, $j=1,...,49$, $m_{(50)}=25$ ($\sum_{j=1}%
^{50}m_{(j)}=n$). Since the clusters are mutually independent and there is
(possibly) correlation inside each cluster, we consider three possible
distributions for%
\[
(n_{11(j)},n_{1(j)}-n_{11(j)},n_{21(j)},n_{2(j)}-n_{21(j)},\ldots
,n_{12,1(j)},n_{12(j)}-n_{12,1(j)})^{T}
\]
corresponding to the $j$-th cluster (column, in Table \ref{EleNirLea_tab:2}),
$j=1,...,J=50$:

\begin{itemize}
\item Dirichlet-multinomial with parameters $(m_{(j)};\rho,\boldsymbol{p}%
\left(  \boldsymbol{\beta}\right)  )$, with $\rho\in\{\frac{1}{10}%
(i-1)\}_{i=1}^{10}$;

\item Random-clumped with parameters $(m_{(j)};\rho,\boldsymbol{p}\left(
\boldsymbol{\beta}\right)  )$, with $\rho\in\{\frac{1}{10}(i-1)\}_{i=1}^{10}$;

\item $m_{(j)}$-inflated with parameters $(m_{(j)};\rho,\boldsymbol{p}\left(
\boldsymbol{\beta}\right)  )$, with $\rho\in\{\frac{1}{10}(i-1)\}_{i=1}^{10}$.
\end{itemize}

For details about these distributions see Alonso et al. \cite{alonso}. The values of
interest for the sample are%
\[
n_{i1}=\sum_{j=1}^{50}n_{11(j)},\quad i=1,...,I\quad\text{and\quad}n_{i}%
=\sum_{j=1}^{50}(n_{11(j)}+n_{12(j)}),\quad i=1,...,I.
\]

\begin{table}[htbp]  \tabcolsep2.8pt  \centering
\caption{Theoretical values of $\boldsymbol{p}\left( \boldsymbol{\beta
}\right)$ in the simulation study.\label{EleNirLea_tab:1}}%
\begin{tabular}
[c]{ccccccccccccc}\hline
$i$ & $1$ & $2$ & $3$ & $4$ & $5$ & $6$ & $7$ & $8$ & $9$ & $10$ & $11$ &
$I=12$\\\hline
$\frac{N_{i}}{N}$ & $\frac{10}{1299}$ & $\frac{63}{1299}$ & $\frac{110}{1299}
$ & $\frac{14}{1299}$ & $\frac{35}{1299}$ & $\frac{281}{1299}$ & $\frac
{40}{1299}$ & $\frac{110}{1299}$ & $\frac{185}{1299}$ & $\frac{204}{1299}$ &
$\frac{196}{1299}$ & $\frac{51}{1299}$\\\rule{0pt}{2.85ex}
$\pi\left(  \boldsymbol{x}_{i}^{T}\boldsymbol{\beta}\right)  $ & $\frac{2}%
{10}$ & $\frac{38}{63}$ & $\frac{65}{110}$ & $\frac{6}{14}$ & $\frac{29}{35}$
& $\frac{188}{281}$ & $\frac{17}{40}$ & $\frac{56}{110}$ & $\frac{105}{185}$ &
$\frac{78}{204}$ & $\frac{93}{196}$ & $\frac{21}{51}$\\\rule{0pt}{2.85ex}
$\frac{N_{i}}{N}\pi\left(  \boldsymbol{x}_{i}^{T}\boldsymbol{\beta}\right)  $
& $\frac{2}{1299}$ & $\frac{38}{1299}$ & $\frac{65}{1299}$ & $\frac{6}{1299}$
& $\frac{29}{1299}$ & $\frac{188}{1299}$ & $\frac{17}{1299}$ & $\frac
{56}{1299}$ & $\frac{105}{1299}$ & $\frac{78}{1299}$ & $\frac{93}{1299}$ &
$\frac{21}{1299}$ \rule{0pt}{2.85ex}\\ \hline 
\end{tabular}

\end{table}%

\begin{table}[htbp]  \tabcolsep2.8pt  \centering
\caption{Scheme of a correlated sample generation through
clusters.\label{EleNirLea_tab:2}}%
\begin{tabular}
[c]{ccccccccc}\hline
$i$ & $j$ & $1$ & $2$ & $\cdots$ & $j$ & $\cdots$ & $J=50$ & sample\\\hline
$1$ &
\begin{tabular}
[c]{c}%
$k=1$\\
$k=2$%
\end{tabular}
&
\begin{tabular}
[c]{c}%
$n_{11(1)}$\\
$n_{12(1)}$%
\end{tabular}
&
\begin{tabular}
[c]{c}%
$n_{11(2)}$\\
$n_{12(2)}$%
\end{tabular}
& $\cdots$ &  & $\cdots$ &
\begin{tabular}
[c]{c}%
$n_{11(50)}$\\
$n_{12(50)}$%
\end{tabular}
&
\begin{tabular}
[c]{c}%
$n_{11}$\\
$n_{12}$%
\end{tabular}
\\
$2$ &
\begin{tabular}
[c]{c}%
$k=1$\\
$k=2$%
\end{tabular}
&
\begin{tabular}
[c]{c}%
$n_{21(1)}$\\
$n_{22(1)}$%
\end{tabular}
&
\begin{tabular}
[c]{c}%
$n_{21(2)}$\\
$n_{22(2)}$%
\end{tabular}
& $\cdots$ &  & $\cdots$ &
\begin{tabular}
[c]{c}%
$n_{21(50)}$\\
$n_{22(50)}$%
\end{tabular}
&
\begin{tabular}
[c]{c}%
$n_{21}$\\
$n_{22}$%
\end{tabular}
\\
$\vdots$ &  & $\vdots$ & $\vdots$ & $\ddots$ &  &  & $\vdots$ & $\vdots$\\
$i$ &
\begin{tabular}
[c]{c}%
$k=1$\\
$k=2$%
\end{tabular}
&
\begin{tabular}
[c]{c}%
$n_{i1(1)}$\\
$n_{i2(1)}$%
\end{tabular}
&
\begin{tabular}
[c]{c}%
$n_{i1(2)}$\\
$n_{i2(2)}$%
\end{tabular}
&  &
\begin{tabular}
[c]{c}%
$n_{i1(j)}$\\
$n_{i2(j)}$%
\end{tabular}
&  &
\begin{tabular}
[c]{c}%
$n_{i1(50)}$\\
$n_{i2(50)}$%
\end{tabular}
&
\begin{tabular}
[c]{c}%
$n_{i1}$\\
$n_{i2}$%
\end{tabular}
\\
$\vdots$ &  & $\vdots$ & $\vdots$ &  &  & $\ddots$ & $\vdots$ & $\vdots$\\
$I=12$ &
\begin{tabular}
[c]{c}%
$k=1$\\
$k=2$%
\end{tabular}
&
\begin{tabular}
[c]{c}%
$n_{12,1(1)}$\\
$n_{12,2(1)}$%
\end{tabular}
&
\begin{tabular}
[c]{c}%
$n_{12,1(2)}$\\
$n_{12,2(2)}$%
\end{tabular}
& $\cdots$ &  & $\cdots$ &
\begin{tabular}
[c]{c}%
$n_{12,1(50)}$\\
$n_{12,2(50)}$%
\end{tabular}
&
\begin{tabular}
[c]{c}%
$n_{12,1}$\\
$n_{12,2}$%
\end{tabular}
\\\hline
&  & $m_{(1)}$ & $m_{(2)}$ & $\cdots$ & $m_{(j)}$ & $\cdots$ & $m_{(50)}$ &
$n$\\\hline
\end{tabular}
$\ $

\end{table}%

Notice that the assumptions of Theorem 1 are held. In addition:

\begin{itemize}
\item If $\rho=0$ (multinomial distribution within each cluster), then
$\boldsymbol{V}$ is a diagonal matrix since the elements of
$\widehat{\boldsymbol{p}}$ are uncorrelated. In this case, we obtain MLEs and
M$\phi$Es.

\item If $\rho>0$, then $\boldsymbol{V}$ is not a diagonal matrix since the
elements of $\widehat{\boldsymbol{p}}$ are correlated. In this case, we obtain
PMLEs and PM$\phi$Es.
\end{itemize}

In these scenarios, the root of the mean square error (RMSE) for the PMCREs of  $\boldsymbol{\beta}$ are studied, considering different values of the tuning parameter $\lambda \in \{0,2/3,1,2 \}$. Note that when $\lambda=0$, the corresponding PMCRE of $\boldsymbol{\beta}$ is equal to the PMLE. 

Results of the simulation study with 2,000 samples are shown in Figure \ref{EleNirLea_fig:1}. As expected from a theoretical point of view, the RMSE increases as $\rho$ increases. With independence to the distribution considered, estimators corresponding to $\lambda \in \{2/3,1,2\}$ present a better performance than the PMLE ($\lambda=0$). This difference becomes more considerable for large values of  $\rho$. 
\begin{figure}[th]
 \includegraphics[scale=0.34]{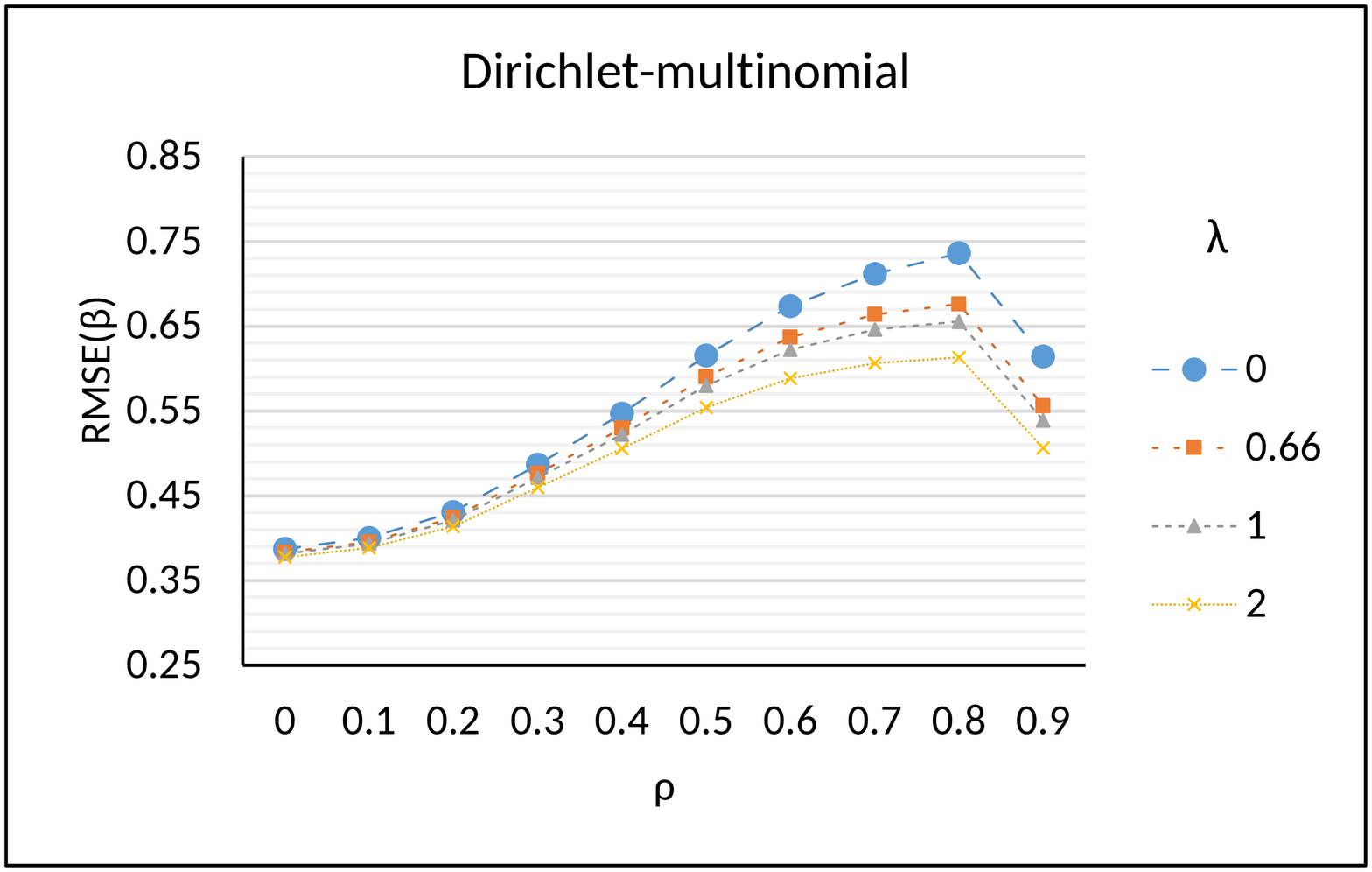}
\\
\includegraphics[scale=0.34]{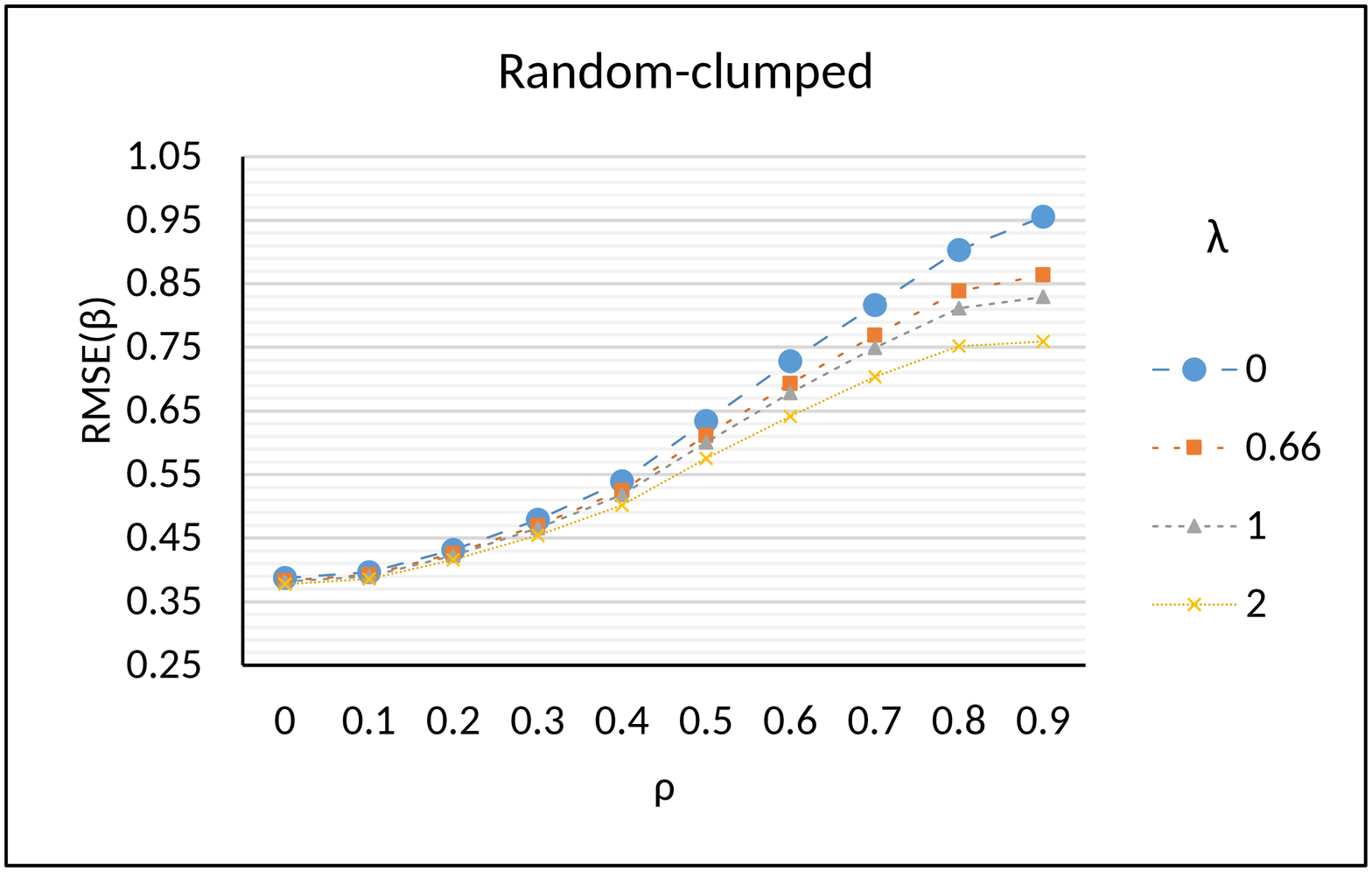}
\\
\includegraphics[scale=0.34]{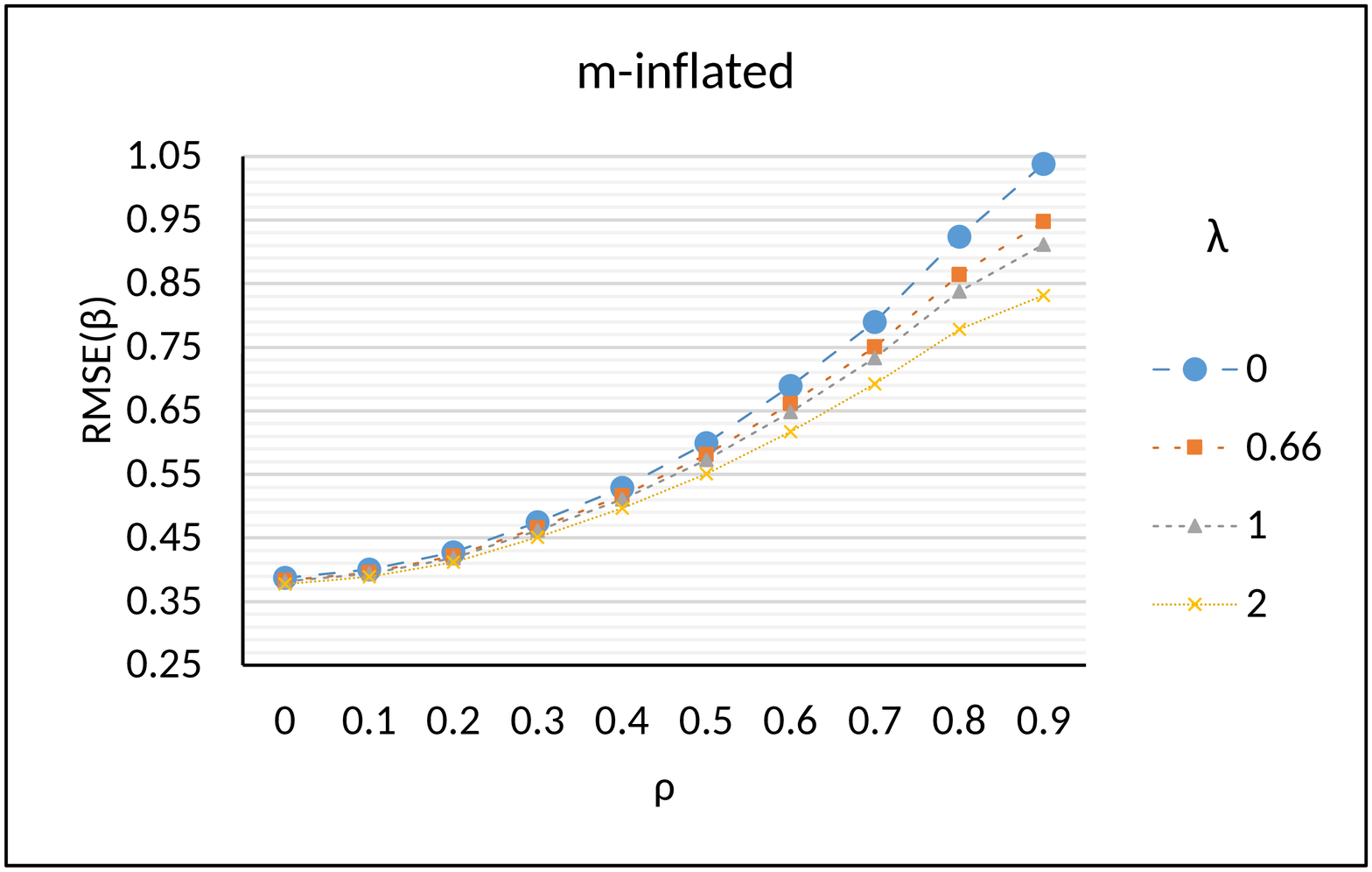}
\\
\caption{RMSEs for PMCREs of $\boldsymbol{\beta}$ with Dirichlet-multinomial (above), Random-clumped  (middle) and m-inflated (below) distributions.}
\label{EleNirLea_fig:1}       
\end{figure}

\section{Concluding remarks}
\label{EleNirLea_sec:5}
In this paper we have considered the problem of estimating the parameters of the logistic regression model for sample survey data, introducing the family of the PM$\phi$Es that contains as a particular case the PMLE. A simulation study is carried out in order to see that there are PM$\phi$Es that have a better behaviour than the PMLE in relation to the mean square error.

\clearpage

\begin{acknowledgement}
This research is partially supported by Grants MTM2015-67057-P(MINECO/FEDER) and ECO2015-66593, both from Ministerio de Economia y Competitividad (Spain).
\end{acknowledgement}

\bibliographystyle{spbasic}

\begin{thebibliography}{99.}

\bibitem {alonso} Alonso-Revenga, J.M., Mart\'{\i}n, N. \& Pardo, L. (2016): New improved
estimators for overdispersion in models with clustered multinomial data and
unequal cluster sizes. Statistics and Computing, doi:10.1007/s11222-015-9616-z.

\bibitem {castilla}Castilla, E., Martin, N. and Pardo, L. (2016). Pseudo
minimum phi-divergence estimator for multinomial logistic regression with
complex sample design. https://arxiv.org/abs/1606.01009.

\bibitem {p1}Cressie, N. and L. Pardo (2003). Phi-divergencia statistic. In
A.H. El-Shaarawi and W.W. Piegorsch (Eds.), Encyclopedia of Environmetrics,
Vol. 3 (pp. 1551--1555). New York: Wiley.

\bibitem {molina}Molina, E. A., Skinner, C. J. (1992). Pseudo-likelihood and quasi-likelihood estimation for complex sampling schemes. Computational Statistics \& Data Analysis 13:395-405.

\bibitem {p3}Pardo, J. A., Pardo, M. C. and Pardo, L. (2005). Minimum $\phi
$-divergence estimator in logistic regression models. Statistical
Papers 47:91-108.

\bibitem {p6}Pardo, J. A., Pardo, M. C. and Pardo, L. (2006). Testing in
logistic regression models based on $\phi$-divergence measures.
Journal of Statistical Planning and Inference 136: 982--1006.

\bibitem {p2}Pardo, L. (2006). Statistical Inference Based on
Divergence Measures. Chapman \& Hall/CRC.

\bibitem {r}Roberts, G., Rao, J.N.K. and Kumer, S. (1987). Logistic Regression
Analysis of Sample Survey Data. Biometrika, 74:1--12.
\end{thebibliography}

\end{document}